%% file: paper.tex
\documentclass[conference]{IEEEtran}
\IEEEoverridecommandlockouts

\usepackage[T1]{fontenc}
\usepackage[utf8]{inputenc}
\usepackage{amsmath,amssymb,amsfonts}
\usepackage{algorithm,algpseudocode}
\usepackage{graphicx}
\usepackage{cite}
\usepackage{bm}
\usepackage{xcolor}
\usepackage{siunitx}
\usepackage{url}
\usepackage{stmaryrd}
\usepackage{microtype}
\usepackage{afterpage}
\usepackage[letterpaper,top=0.75in,bottom=1in,left=0.7in,right=0.7in]{geometry}

\newtheorem{theorem}{Theorem}
\newtheorem{proposition}{Proposition}

\newtheorem{remark}{Remark}

\newcommand{\calW}{\mathcal{W}}
\newcommand{\calK}{\mathcal{K}}
\newcommand{\calL}{\mathcal{L}}

\newcommand{\calC}{\mathcal{C}}
\newcommand{\calD}{\mathcal{D}}
\newcommand{\calI}{\mathcal{I}}
\newcommand{\calN}{\mathcal{N}}
\newcommand{\RR}{\mathbb{R}}
\newcommand{\GF}{\mathrm{GF}}

\newcommand{\expc}[2]{\mathbb{E}_{#1}\!\left[#2\right]}
\newcommand{\expcs}[2]{\mathbb{E}_{#1}[#2]}
\newcommand{\pr}[1]{\Pr\!\left\{#1\right\}}
\DeclareMathOperator{\dH}{d_H}

\begin{document}

\title{Variational Probabilistic Quantization for\\ Secret Key Generation}

\author{\IEEEauthorblockN{Xinyang Li, Vlad C. Andrei, Peter J. Gu, Yiqi Chen, Ullrich J.\ M\"onich, Holger Boche}
\IEEEauthorblockA{Chair of Theoretical Information Technology, Technical University of Munich, Munich, Germany\\
Email: \{xinyang.li, vlad.andrei, peter.gu, yiqi.chen, moenich, boche\}@tum.de}
\thanks{This work was supported by the Bavarian Ministry of Economic Affairs, Regional Development and Energy within the project ``6G Future Lab Bavaria''; by the project ``Next Generation AI Computing (gAIn)'', funded by the Bavarian Ministry of Science and the Arts and the Saxon Ministry for Science, Culture, and Tourism; and by the Federal Ministry of Research, Technology and Space (BMFTR) under the programme CommUnity (16KISS013) and the programme ``Souver\"an. Digital. Vernetzt.'', joint project 6G-life (16KIS2414).}
\thanks{\textcopyright\ 2026 IEEE. Personal use of this material is permitted. Permission from IEEE must be obtained for all other uses, in any current or future media, including reprinting/republishing this material for advertising or promotional purposes, creating new collective works, for resale or redistribution to servers or lists, or reuse of any copyrighted component of this work in other works.}
}

\maketitle

\begin{abstract}
Secret key generation from correlated observations at Alice and Bob, in the presence of an eavesdropper Eve, underpins physical-layer security. Classical pipelines quantize by hand, amplify privacy afterwards, and optimize no objective tied to a key rate. We propose Variational Probabilistic Quantization (VPQ): neural encoders that map the correlated sources directly into a discrete key alphabet, trained by a variational adversarial objective whose entropy, mismatch, and leakage terms match the three terms of the one-way secret key rate. A linear code-offset secure sketch then reconciles the encoder outputs into an identical key without a separate privacy amplification step. We prove that the VPQ losses lower-bound the one-way secret key capacity of the induced source, and derive in closed form the optimal worst-case key rate over the source class of a given alphabet size and mismatch probability, attained by finite-field linear sketches. On Gaussian fading channels, VPQ leaks less to a correlated eavesdropper than one classical and two recent learning-based baselines, and Reed--Solomon reconciliation operates within the predicted finite-blocklength gap.
\end{abstract}


\input{sections/intro}
\input{sections/sysmodel}
\input{sections/method}
\input{sections/theory}
\input{sections/experiments}
\input{sections/conclusion}

\bibliographystyle{IEEEtran}
\bibliography{refs}

\end{document}

%% file: sections/intro.tex
\section{Introduction}
\label{sec:intro}

Secret keys derived from the wireless channel are a basic building block of physical-layer security~\cite{wyner1975wire,nguyen2021security}. The dominant pipeline for physical-layer key (PLK) generation~\cite{wang2011fast,zhang2016key} is fixed in form: channel probing, hand-designed scalar quantization based on level-crossing or cumulative-distribution-function (CDF) rules, information reconciliation, and privacy amplification. Such quantizers act on the marginal of a single scalar observable, ignore the joint structure of Alice's and Bob's observations, and treat secrecy as an after-the-fact decorrelation step rather than as a design objective. As a consequence, their secrecy degrades sharply once Eve becomes correlated with the legitimate channel, which is precisely the regime where PLK generation must provide protection.

Recent learning-based methods~\cite{han2020physical,guo2025physical,zhou2023daae} replace the hand-designed quantizer with an autoencoder that maps the observations into reciprocal continuous latent features, followed by a separate, non-differentiable quantization step; secrecy is either deferred to privacy amplification or, in~\cite{zhou2023daae}, enforced by a domain-adversarial branch. The discretization still sits outside the training loop, and no objective is tied to an information-theoretic key rate.

We map the correlated sources directly into a discrete key alphabet in a single end-to-end differentiable step, so that uniformity, agreement, and secrecy are governed by the training objective rather than recovered post hoc. We propose Variational Probabilistic Quantization (VPQ): probabilistic neural-network encoders trained by a variational adversarial objective whose entropy, mismatch, and leakage terms are optimized against a learnable predictor at Eve~\cite{poole2019variational,cheng2020club}. A code-offset secure sketch~\cite{dodis2004fuzzy} then reconciles the encoder outputs into an identical key. Because the adversarial objective already drives the leakage to Eve toward zero, no separate privacy amplification step is required. The construction remains effective under a correlated Eve.

\emph{Contributions.}
(i) A two-stage framework that combines VPQ with a $q$-ary linear code-offset secure sketch and produces, through a single variational adversarial loss, near-uniform and agreeing discrete random variables with low leakage to Eve.
(ii) Two information-theoretic guarantees, valid in particular under a correlated Eve: the VPQ training losses lower-bound the one-way secret key capacity of the induced source, and the optimal worst-case one-way reconciliation rate over the i.i.d.\ source class with a given alphabet size and mismatch probability admits a closed form attained by $q$-ary linear code-offset sketches; a bounded-leakage proposition caps the key's leakage rate by the residual leakage.
(iii) An empirical study on Gaussian fading sources under absent, uncorrelated, and correlated Eve, benchmarked against the conventional key generation (Conv.\ KG) pipeline of~\cite{zhang2016key} and two recent learning-based methods, the mutual-information-driven autoencoder (MIAE)~\cite{guo2025physical} and the domain-adversarial autoencoder (DAAE)~\cite{zhou2023daae}: VPQ delivers usable keys without a privacy amplification stage even against a correlated Eve, whereas all three baselines leak up to 2.2 bits per symbol.

\emph{Notation.} Random variables are written in upper case and their realizations in lower case; $\bm x$ denotes a vector. $H(\cdot)$ and $I(\cdot\,;\cdot)$ are the entropy and the mutual information in bits, $h_2(\cdot)$ the binary and $h_q(\cdot)$ the $q$-ary entropy function, and $D(\cdot\|\cdot)$ the Kullback--Leibler divergence. Alice and Bob quantize their observations into $W$ and $V$ on a common alphabet $\calW$ of size $q=|\calW|$, with mismatch rate $p_e=\pr{W\neq V}$; $\GF(q)$ is the finite field of order $q$, $\calD_{p_e}$ the i.i.d.\ source class of Theorem~\ref{thm:univ}, $B$ the batch size, $\alpha$ the averaging factor, and $\lambda_1$, $\lambda_2$ the loss weights, all introduced in Section~\ref{sec:vpq}.

%% file: sections/sysmodel.tex
\section{System Model}
\label{sec:sys}

\begin{figure*}[!t]
\centering
\includegraphics[width=0.68\textwidth]{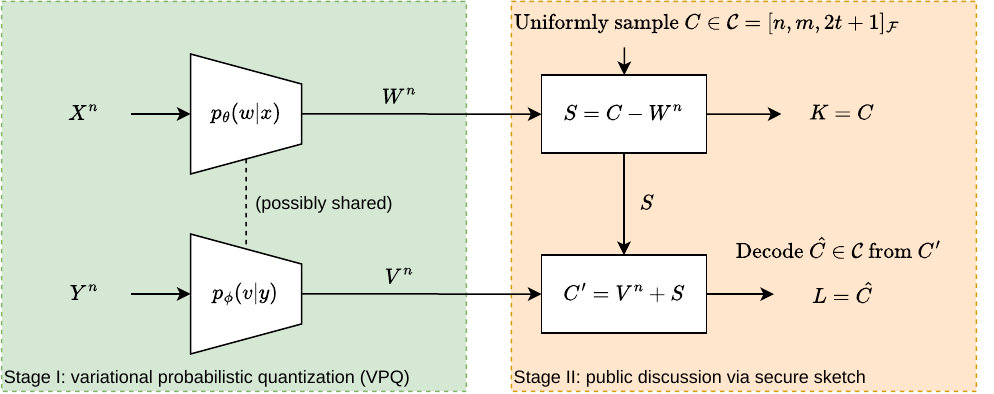}
\caption{Two-stage framework. Stage I (VPQ): probabilistic encoders $p_\theta(w|x)$ and $p_\phi(v|y)$ map correlated observations $X^n,Y^n$ elementwise to discrete sequences $W^n,V^n$ on a common $q$-ary alphabet $\calW$. Stage II: a $q$-ary linear code-offset secure sketch turns $(W^n,V^n)$ into an identical key $K=L$ via a single one-way public message $S$, with no privacy amplification step.}
\label{fig:overview}
\end{figure*}

Alice, Bob, and Eve observe correlated sequences $\{(X_i,Y_i,Z_i)\}_{i=1}^{n}$ drawn i.i.d.\ from a joint distribution $p(x,y,z)$, with $Z=\varnothing$ if Eve is absent; Fig.~\ref{fig:overview} shows the two stages. In the one-way public-discussion setting~\cite{ahlswede1993common,maurer1993secret}, Alice transmits a public message $M=\Phi(X^n)$ and the two parties extract random variables
\begin{equation}
K = f(X^n),\qquad L = g(Y^n,M),
\end{equation}
on a common alphabet $\calK$; the mappings $f,g,\Phi$ may be deterministic or stochastic. The pair $(K,L)$ is a secret key if, for every $\epsilon>0$ and sufficiently large $n$,
\begin{align}
\pr{K\neq L}&<\epsilon, \label{eq:agree}\\
\tfrac{1}{n}\log|\calK|-\tfrac{1}{n}H(K)&<\epsilon, \label{eq:unif}\\
\tfrac{1}{n}I(K;Z^n,M)&<\epsilon. \label{eq:leak}
\end{align}
The supremum of achievable rates $\tfrac{1}{n}H(K)$ defines the one-way secret key (SK) capacity~\cite{maurer1993secret,ahlswede1993common,csiszar2011information},
$C_{\mathrm{SK}}(X,Y|Z)=\max\, I(T;Y|U)-I(T;Z|U)$,
where the maximum is over all auxiliary $(U,T)$ satisfying the Markov chain $U{-}T{-}X{-}(Y,Z)$.

Although $C_{\mathrm{SK}}$ is well characterized, practical extraction is open when $p(x,y,z)$ is unknown and high-dimensional. We address this through a two-stage framework: a variational stage that learns encoders at Alice and Bob mapping $(X,Y)$ into discrete random variables $(W,V)$ on a $q$-ary alphabet $\calW$, and a reconciliation stage that uses linear code-offset sketches to convert $(W^n,V^n)$ into an identical key satisfying~\eqref{eq:agree}--\eqref{eq:leak}.

%% file: sections/method.tex
\section{Two-Stage Framework}
\label{sec:vpq}

\subsection{Variational Probabilistic Quantization}
\label{subsec:vpq}

VPQ employs two probabilistic encoders, $p_\theta(w|x)$ and $p_\phi(v|y)$, mapping observations elementwise to discrete random variables $(W,V)$ on the common $q$-ary alphabet $\calW$. The last two layers of each encoder are a linear projection followed by a softmax of dimension $q$; the encoders may share parameters when $X$ and $Y$ are statistically symmetric. Three loss terms jointly drive $(W,V)$ toward high agreement, near-uniform marginals, and low mutual information with Eve's observation.

\emph{Mismatch loss.} Maximizing the agreement rate $\pr{W{=}V}$ means maximizing $\sum_{w'\in\calW}\Pr\{W{=}w',V{=}w'\}$. With $\bm w$ and $\bm v$ the one-hot embeddings of $W$ and $V$, agreement is the inner product $\bm w^\top \bm v$, so $\pr{W{=}V \mid x,y}=\expc{p_\theta(w|x)p_\phi(v|y)}{\bm w^\top\bm v}$. Conditional independence of $W$ and $V$ given $(X,Y)$ factorizes this expectation,
\begin{equation}
\pr{W{=}V \mid x,y}
= \sum_{w'\in\calW} p_\theta(w'|x)\,p_\phi(w'|y),
\end{equation}
yielding the empirical mismatch loss over a batch $\{(x_i,y_i)\}_{i=1}^{B}$,
\begin{equation}\label{eq:lmr}
\calL_{\mathrm{MR}} = -\frac{1}{B}\sum_{i=1}^{B}\sum_{w'\in\calW} p_\theta(w'|x_i)\,p_\phi(w'|y_i),
\end{equation}
which is fully differentiable in $(\theta,\phi)$ and drives both encoders toward a common (one-hot) output on each paired input.

\emph{Uniformity loss.} The mismatch loss alone admits a degenerate solution $p_\theta(w|x)=p_\phi(w|y)=\bm 1\{w{=}w_0\}$ for some fixed $w_0$, which violates the uniformity requirement~\eqref{eq:unif}. We therefore push $H(W)$ and $H(V)$ toward $\log q$. When $|\calW|$ exceeds the batch size, batch-wise marginal estimates are noisy, so we maintain exponential moving averages (EMA) $p_t(w)$ of the marginal $\Pr\{W=w\}$ and $p_t(v)$ of $\Pr\{V=v\}$ across training steps,
\begin{equation}\label{eq:ema}
p_t(w)=\alpha\,p_{t-1}(w)+\frac{1-\alpha}{B}\sum_{i=1}^{B}p_\theta(w|x_i),
\end{equation}
with $0\le\alpha<1$, $p_{t-1}$ detached from the gradient, and $p_t(v)$ defined analogously. The uniformity loss is
\begin{equation}\label{eq:lent}
\calL_{\mathrm{ENT}} = -\frac{1}{2(1-\alpha)}\bigl(\hat H(W)+\hat H(V)\bigr),
\end{equation}
where $\hat H(W)=-\sum_w p_t(w)\log p_t(w)$ and $\hat H(V)$ is defined symmetrically; the factor $(1-\alpha)^{-1}$ rescales the current-batch gradient contribution to match the un-EMA estimator, since only the last term in~\eqref{eq:ema} contributes to the gradient.

\emph{Leakage loss.} When Eve is absent or uncorrelated with $(X,Y)$, the data processing inequality on $W{-}X{-}Z$ gives $I(W;Z)\le I(X;Z)=0$ and the loss reduces to $\calL_{\mathrm{AB}}=\calL_{\mathrm{ENT}}+\lambda_1\calL_{\mathrm{MR}}$, $\lambda_1>0$. When $Z$ is correlated, $I(W;Z)$ is intractable, and we sandwich it variationally against a learnable predictor $p_\psi(w|z)$ at Eve. The Barber--Agakov variational lower bound (VLB)~\cite{poole2019variational}, $I_{\mathrm{VLB}}:=H(W)+\expcs{p(w,z)}{\log p_\psi(w|z)}\le I(W;Z)$, is tight at $p_\psi(w|z)=p(w|z)$; the contrastive log-ratio upper bound (CLUB)~\cite{cheng2020club} provides a variational upper bound (VUB), $I_{\mathrm{VUB}}:=\expcs{p(w,z)}{\log p_\psi(w|z)}-\expcs{p(w)p(z)}{\log p_\psi(w|z)}\ge I(W;Z)$ whenever $p_\psi\approx p(w|z)$. With the joint expectation replaced by a minibatch average via $W\perp Z\mid X$ and the $H(W)$ term of $I_{\mathrm{VLB}}$ dropped (a constant under fixed $(\theta,\phi)$, separately controlled by $\calL_{\mathrm{ENT}}$), the empirical batch estimators are
\begin{align}
\calI_{\mathrm{VLB}} &= \tfrac{1}{B}\textstyle\sum_{i,w} p_\theta(w|x_i)\log p_\psi(w|z_i), \label{eq:vlb}\\
\calI_{\mathrm{VUB}} &= \calI_{\mathrm{VLB}} - \tfrac{1}{B^2}\textstyle\sum_{i,j,w} p_\theta(w|x_i)\log p_\psi(w|z_j), \label{eq:vub}
\end{align}
which respectively underestimate and (after a successful adversarial step) overestimate $I(W;Z)$. The encoders $(\theta,\phi)$ minimize
\begin{equation}\label{eq:total}
\calL = \calL_{\mathrm{ENT}} + \lambda_1 \calL_{\mathrm{MR}} + \lambda_2 \calI_{\mathrm{VUB}},
\end{equation}
where the leakage weight $\lambda_2 = \|\nabla_{\theta_L}\calL_{\mathrm{AB}}\|_2 / (\|\nabla_{\theta_L}\calI_{\mathrm{VUB}}\|_2 + \delta)$ adaptively balances the two gradient norms ($\theta_L$ is the last linear layer, $\delta=10^{-7}$). A direct calculation gives $I_{\mathrm{VUB}}-I_{\mathrm{VLB}}=\expcs{p(z)}{D(p(w)\|p_\psi(w|z))}\ge 0$, so the sandwich tightens as $p_\psi\to p(w|z)$.

Algorithm~\ref{alg:vpq} summarizes training. Under a correlated Eve, the $\psi$-update is disabled during a warm-up phase that updates $(\theta,\phi)$ only; otherwise $p_\psi$ overfits a random encoder and the leakage signal is uninformative. After warm-up, $\psi$- and $(\theta,\phi)$-updates alternate one-to-one.

\begin{algorithm}[!t]
\caption{VPQ training}\label{alg:vpq}
\begin{algorithmic}[1]
\For{each training step $t$}
  \State Sample batch $\{(x_i,y_i,z_i)\}_{i=1}^{B}$; compute $p_\theta(w|x_i),p_\phi(v|y_i)$
  \State Compute $\calL_{\mathrm{MR}}$~\eqref{eq:lmr}; update $p_t(w),p_t(v)$~\eqref{eq:ema}; compute $\calL_{\mathrm{ENT}}$
  \State $\calL_{\mathrm{AB}}\!\leftarrow\!\calL_{\mathrm{ENT}}+\lambda_1\calL_{\mathrm{MR}}$
  \If{Eve is present}
    \If{$\psi$-update step}
      \State Compute $p_\psi(w|z_i)$, $\calI_{\mathrm{VLB}}$~\eqref{eq:vlb}; update $\psi$ to maximize $\calI_{\mathrm{VLB}}$
    \EndIf
    \If{$(\theta,\phi)$-update step}
      \State Compute $\calI_{\mathrm{VUB}}$~\eqref{eq:vub}, $\lambda_2$; update $(\theta,\phi)$ to minimize $\calL_{\mathrm{AB}}+\lambda_2\calI_{\mathrm{VUB}}$
    \EndIf
  \Else
    \State Update $(\theta,\phi)$ to minimize $\calL_{\mathrm{AB}}$
  \EndIf
\EndFor
\end{algorithmic}
\end{algorithm}

\begin{remark}[Lautum diagnostic]\label{rem:lautum}
With a well-trained predictor $p_\psi(w|z)\approx p(w|z)$, the theoretical gap $I_{\mathrm{VUB}}-I_{\mathrm{VLB}}$ converges to the lautum information $L(W;Z):=\expcs{p(z)}{D(p(w)\|p(w|z))}$, which vanishes if and only if $I(W;Z)=0$~\cite{palomar2008lautum}. Since $\calI_{\mathrm{VLB}}$ in~\eqref{eq:vlb} omits the additive $H(W)$ in $I_{\mathrm{VLB}}$, the empirical gap $\calI_{\mathrm{VUB}}-\calI_{\mathrm{VLB}}$ converges to $H(W)+L(W;Z)$. In the uniform regime enforced by $\calL_{\mathrm{ENT}}$ this limit equals $\log q$ exactly when the leakage vanishes, so observing $\calI_{\mathrm{VUB}}-\calI_{\mathrm{VLB}}\to\log q$ during training is a sufficient empirical signature of zero leakage.
\end{remark}

\subsection{Code-Offset Reconciliation}
\label{subsec:reconcile}

Since $p_e$ is generally nonzero, the probability that $W^n=V^n$ decays exponentially in $n$, so a public-discussion step is needed. We use the code-offset secure sketch~\cite{dodis2004fuzzy} over the finite field $\GF(q)$~\cite{roth2006introduction} with an $[n,m,2t+1]_q$ linear code $\calC$ of blocklength $n$, dimension $m$, and minimum Hamming distance $2t+1$: Alice samples $K$ uniformly from $\calC$, sends $S=K-W^n$, and Bob computes $L=\hat K=\arg\min_{c\in\calC}\dH(c,V^n+S)$, where $\dH(\cdot,\cdot)$ denotes Hamming distance. Whenever $\dH(W^n,V^n)\le t$, $\hat K=K$. Each block produces $H(K)/n=(m/n)\log q$ key bits.

%% file: sections/theory.tex
\section{Information-Theoretic Guarantees}
\label{sec:theory}

We connect the VPQ losses to the achievable secret key rate in two steps. The first relates the discrete source $(W,V,Z)$ induced by the encoders to the one-way SK capacity. The second characterizes the optimal universal key rate that code-offset reconciliation can attain on top of any encoder with given alphabet size and mismatch probability.

\subsection{Achievable Rate of the Induced Source}

Define the $q$-ary entropy $h_q(p):=[h_2(p)+p\log(q-1)]/\log q$.

\begin{theorem}[VPQ achievable rate]\label{thm:vpq}
For $W,V$ on a $q$-ary alphabet with mismatch probability $p_e=\pr{W\neq V}$ and Eve's observation $Z$,
\begin{equation}\label{eq:thm1}
C_{\mathrm{SK}}(W,V|Z)\;\ge\;H(W)-h_2(p_e)-p_e\log(q-1)-I(W;Z),
\end{equation}
and equality $C_{\mathrm{SK}}(W,V|Z)=I(W;V)$ holds whenever $I(W;Z)=0$.
\end{theorem}

\begin{IEEEproof}
The one-way SK capacity formula~\cite{maurer1993secret,ahlswede1993common} reads
$C_{\mathrm{SK}}=\max_{U-T-W-(V,Z)}I(T;V|U)-I(T;Z|U)$.
Setting $U\equiv\mathrm{const}$ and $T=W$ gives $C_{\mathrm{SK}}\ge I(W;V)-I(W;Z)$. Fano's inequality with $\hat W(V)=V$ and error probability $p_e$ yields $H(W|V)\le h_2(p_e)+p_e\log(q-1)$, so $I(W;V)=H(W)-H(W|V)\ge H(W)-h_2(p_e)-p_e\log(q-1)$, proving~\eqref{eq:thm1}.

When $I(W;Z)=0$, the Markov chain $U{-}T{-}W{-}(V,Z)$ gives $(U,T)\perp Z\mid W$, which combined with $W\perp Z$ yields $(U,T,W)\perp Z$, so $p(u,t,w,z)=p(u,t,w)p(z)$ and $I(T;Z|U)=0$; the capacity reduces to $\max I(T;V|U)$. With $(U,T){-}W{-}V$, data processing gives $I(T;V|U)\le I(W;V)$, matching the lower bound.
\end{IEEEproof}

Each term on the right of~\eqref{eq:thm1} maps directly onto a VPQ loss: $\calL_{\mathrm{ENT}}$ raises $H(W)$ toward $\log q$; $\calL_{\mathrm{MR}}$ shrinks the Fano penalty $h_2(p_e)+p_e\log(q-1)$; and the adversarial $\calI_{\mathrm{VUB}}$ is a variational surrogate for $I(W;Z)$. Every improvement in the three losses raises a rate the encoder pair can deliver. Because the VPQ encoder defines a stochastic mapping $(X,Y,Z)\mapsto(W,V,Z)$, any one-way reconciliation scheme on the induced source composes with the encoder, and data processing gives $C_{\mathrm{SK}}(X,Y|Z)\ge C_{\mathrm{SK}}(W,V|Z)$.

\begin{figure*}[!t]
\centering
\includegraphics[width=0.78\textwidth]{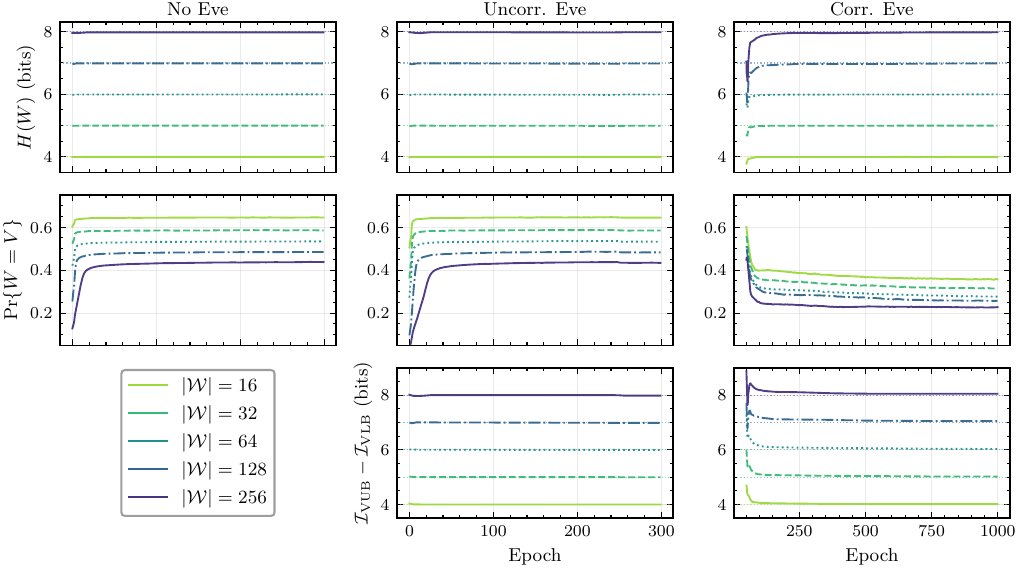}
\caption{VPQ training dynamics on fading channels: encoder entropy $H(W)$ (top), agreement rate $\pr{W=V}$ (middle), and observable leakage gap $\calI_{\mathrm{VUB}}-\calI_{\mathrm{VLB}}$ (bottom) over epochs, for the three Eve scenarios. Horizontal dotted lines mark $\log q$.}
\label{fig:conv}
\end{figure*}

\subsection{Universal Reconciliation Rate}

Given a VPQ output with $W$ uniform on $\GF(q)$ and mismatch probability $p_e$, the appropriate benchmark for reconciliation is the highest one-way rate that can be guaranteed uniformly over all i.i.d.\ source laws compatible with $(q,p_e)$, since the precise law of $(W,V)$ is unknown.

\begin{theorem}[Universal one-way rate]\label{thm:univ}
Let $q$ be a prime power and let $\calD_{p_e}$ be the class of i.i.d.\ triples $(W,V,Z)$ with $W,V$ uniform on $\GF(q)$, $\pr{W\neq V}=p_e<1-1/q$, and $I(W;Z)=0$. The optimal worst-case one-way reconciliation rate over $\calD_{p_e}$ is
\begin{equation}\label{eq:univ-rate}
R^{\star}_{\mathrm{univ}}=(1-h_q(p_e))\log q
\end{equation}
bits per source symbol. Every rate below $R^{\star}_{\mathrm{univ}}$ is achievable, uniformly over $\calD_{p_e}$, by a sequence of $q$-ary linear code-offset secure sketches $S_n=K_n-W^n$ satisfying $\pr{\hat K_n\neq K_n}\to 0$ and $I(K_n;S_n,Z^n)=0$ for every $n$. Operationally, no scheme that is reliable and secure uniformly over $\calD_{p_e}$ can exceed $R^{\star}_{\mathrm{univ}}$; for a specific source law inside $\calD_{p_e}$ a higher rate may be attainable.
\end{theorem}

\begin{IEEEproof}
\emph{Converse.} Consider the source law in $\calD_{p_e}$ in which $Z$ is constant and $P_{V|W}$ is the $q$-ary symmetric channel with crossover $p_e/(q-1)$. For this law, $H(V|W)=h_2(p_e)+p_e\log(q-1)$, $I(W;V)=(1-h_q(p_e))\log q$, and $I(W;Z)=0$, so Theorem~\ref{thm:vpq} gives $C_{\mathrm{SK}}=R^{\star}_{\mathrm{univ}}$. Any scheme that is uniformly reliable and secure over $\calD_{p_e}$ must perform reliably and securely on this particular law as well, and therefore cannot exceed $R^{\star}_{\mathrm{univ}}$.

\emph{Achievability.} Fix $R<R^{\star}_{\mathrm{univ}}$ and pick $\rho$ with $h_q(p_e)<\rho<(\log q-R)/\log q$. Choose $m_n$ with $m_n/n\to\rho$, draw a uniform random parity-check matrix $H_n\in\GF(q)^{m_n\times n}$, set $T_n:=H_nW^{n\top}$, and decode $W^n$ as the unique $\hat w^n$ with $H_n\hat w^{n\top}=T_n$ and Hamming distance $\dH(\hat w^n,V^n)\le n(p_e+\eta)$ for some $\eta>0$ with $\rho>h_q(p_e+\eta)$. The error event splits into (i) $\{\dH(W^n,V^n)>n(p_e+\eta)\}$, vanishing by the law of large numbers, and (ii) syndrome collision on the Hamming ball, bounded by $B_q(n,p_e+\eta)\,q^{-m_n}\to 0$, where $B_q(n,r):=\sum_{j=0}^{\lfloor nr\rfloor}\binom{n}{j}(q-1)^j$ is the volume of the $q$-ary Hamming ball and $\tfrac{1}{n}\log_q B_q(n,p_e+\eta)\to h_q(p_e+\eta)<\rho$~\cite{csiszar2011information,csiszar1982linear}; the bound depends only on $p_e$, hence holds uniformly on $\calD_{p_e}$.

Let $\calC_n:=\ker(H_n)$, draw $K_n$ uniform on $\calC_n$ independently of $(W^n,V^n,Z^n)$, and set $S_n:=K_n-W^n$. Decoding $K_n$ from $(V^n,S_n)$ is equivalent to the syndrome-decoding problem above, since $\hat w^n:=\hat k^n-S_n$ obeys $H_n\hat w^{n\top}=T_n$ and $\dH(\hat w^n,V^n)=\dH(\hat k^n,V^n+S_n)$, so $\pr{\hat K_n\neq K_n}\to 0$ uniformly over $\calD_{p_e}$.

\emph{Secrecy.} Uniformity of $W^n$ over $\GF(q)^n$ and independence of $K_n$ from $W^n$ give $\Pr\{K_n{=}k,S_n{=}s\}=|\calC_n|^{-1}q^{-n}=\Pr\{K_n{=}k\}\Pr\{S_n{=}s\}$, hence $K_n\perp S_n$. With $I(W;Z)=0$, $(K_n,S_n)\perp Z^n$ as well, so $I(K_n;S_n,Z^n)=0$ for every $n$.

\emph{Key rate.} $H(K_n)=\log|\calC_n|\ge(n-m_n)\log q$, giving $\liminf_n n^{-1}H(K_n)\ge(1-\rho)\log q>R$, and $R<R^{\star}_{\mathrm{univ}}$ was arbitrary.
\end{IEEEproof}

The bound~\eqref{eq:univ-rate} matches the encoder-induced rate of Theorem~\ref{thm:vpq} at $I(W;Z)=0$, so reconciliation is not the bottleneck of the two-stage design.

\begin{proposition}[Bounded leakage]\label{prop:bound}
Under the same setup as Theorem~\ref{thm:univ} except that the zero-leakage condition is relaxed to $I(W;Z)\le\varepsilon$ for some $\varepsilon\ge 0$, the same code-offset construction achieves
\begin{itemize}
\item[(i)] reliability $\pr{\hat K_n\neq K_n}\to 0$;
\item[(ii)] rate-$\varepsilon$ leakage $\tfrac{1}{n}I(K_n;S_n,Z^n)\le\varepsilon$ for every $n$ (which reduces to weak secrecy in the limit $\varepsilon\to 0$);
\item[(iii)] key rate: for every $\delta>0$, the construction at target rate $R^{\star}_{\mathrm{univ}}-\delta$ attains $\liminf_n \tfrac{1}{n}H(K_n)\ge(1-h_q(p_e))\log q-\delta$.
\end{itemize}
\end{proposition}

\begin{IEEEproof}
The reliability and rate analyses of Theorem~\ref{thm:univ} use only $p_e$ and the uniformity of $W$, so (i) and (iii) hold verbatim. The same holds for $K_n\perp S_n$, hence $I(K_n;S_n)=0$. Then
\begin{align}
I(K_n;S_n,Z^n)
&= I(K_n;Z^n\mid S_n) = I(W^n;Z^n\mid S_n) \notag\\
&\le I(W^n;Z^n) = nI(W;Z) \le n\varepsilon,
\end{align}
using the bijection $K_n\leftrightarrow W^n=K_n-S_n$ given $S_n$ for the second equality; the bound $I(W^n;Z^n|S_n)\le I(W^n;Z^n)$ holds because $K_n\perp(W^n,Z^n)$ implies $S_n\perp Z^n\mid W^n$, hence the Markov chain $S_n{-}W^n{-}Z^n$ and $I(W^n;Z^n|S_n)=I(W^n;Z^n)-I(S_n;Z^n)\le I(W^n;Z^n)$; the last equality uses the i.i.d.\ assumption. Dividing by $n$ proves (ii).
\end{IEEEproof}

Unlike the conventional pipeline~\cite{wang2011fast,zhang2016key}, our construction delivers the key directly without privacy amplification.

\begin{figure*}[!t]
\centering
\includegraphics[width=0.82\textwidth]{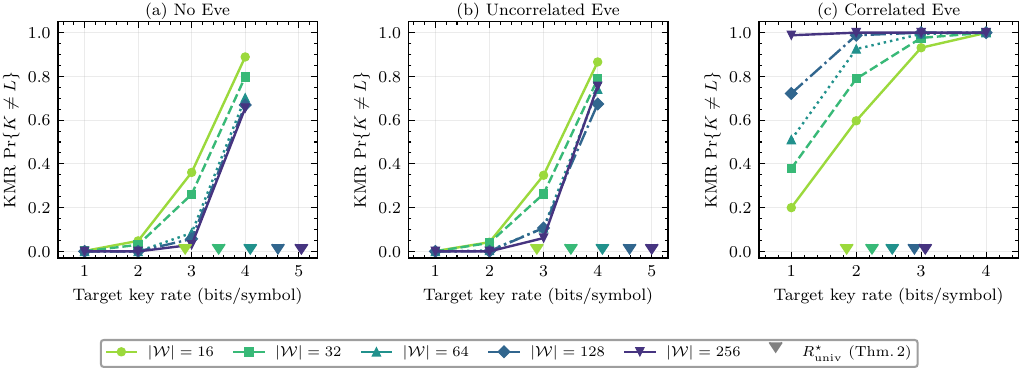}
\caption{Key mismatch rate (KMR) versus target key rate at $\mathrm{SNR}_{\mathrm{ab}}=20$~\si{dB} for VPQ at $q\in\{16,\dots,256\}$ across the three Eve scenarios. Triangles on the $x$-axis mark the per-alphabet universal limit $R^{\star}_{\mathrm{univ}}=(1-h_q(p_e))\log q$ from Theorem~\ref{thm:univ}, evaluated at the measured $p_e$.}
\label{fig:kmr}
\end{figure*}

\begin{figure}[!t]
\centering
\includegraphics[width=0.66\columnwidth]{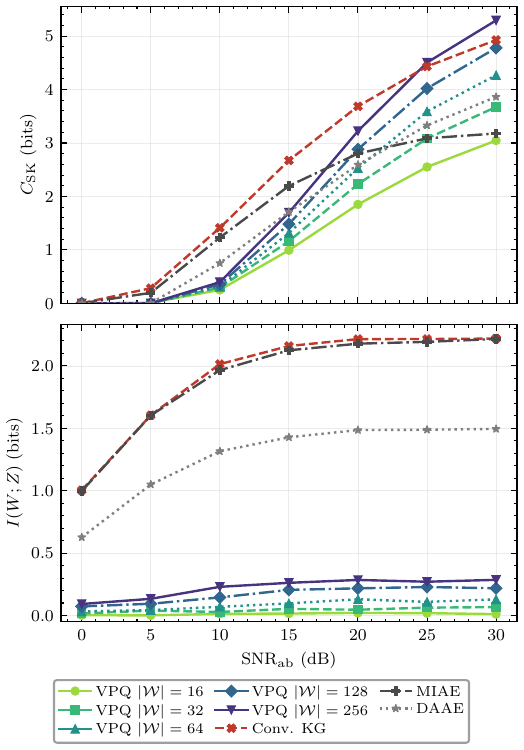}
\caption{Correlated-Eve scenario: achievable secret key rate $I(W;V)-I(W;Z)$ (top) and Eve leakage $I(W;Z)$ (bottom) versus $\mathrm{SNR}_{\mathrm{ab}}$. VPQ keeps $I(W;Z)$ below 0.3 bits for every $q$; the baselines leak up to 2.2 bits per symbol.}
\label{fig:corr}
\end{figure}

%% file: sections/experiments.tex
\section{Experiments on Fading Channels}
\label{sec:exp}

We evaluate VPQ on a canonical PLK setting where Alice and Bob observe correlated Gaussian samples
\begin{equation}
X = H + W_1,\qquad Y = H + W_2,
\end{equation}
with $H\sim\calN(0,P)$, $W_1\sim\calN(0,N_1)$, $W_2\sim\calN(0,N_2)$ independent, and each vector $X,Y\in\RR^{8}$ consists of $d=8$ i.i.d.\ such scalar samples. We consider three Eve scenarios: absent ($Z=\varnothing$); uncorrelated ($Z$ independent Gaussian); and correlated ($Z=H+W_3$, $W_3\sim\calN(0,N_3)$). We set $P=N_3=0$~\si{dBm}.

\emph{Mixed-SNR training.} We adopt the mixed-SNR training of~\cite{guo2025physical}: the per-sample legitimate noise variance is drawn uniformly, $N_1=N_2\sim\mathrm{Uniform}[-30,0]$~\si{dBm}, so that the legitimate signal-to-noise ratio (SNR) $\mathrm{SNR}_{\mathrm{ab}}:=P/N_1$ covers $[0,30]$~\si{dB} at every training step; the trained encoder is tested at specific SNRs. One encoder therefore serves the whole range without retraining or SNR feedback. The noise terms $W_1,W_2$ model the estimation error of the legitimate link.

\emph{Implementation.} Encoder $p_\theta$ is a fully connected network (FCN) with $8$ hidden layers of width $d_{\mathrm{FCN}}=1024$, batch normalization, and rectified-linear-unit (ReLU) activations; predictor $p_\psi$ is an $8$-layer FCN of width $2048$. We set $B=2048$; $\alpha=0.6$ without Eve and $0.9$ with Eve; $\lambda_1=1$ without or with an uncorrelated Eve, and $\lambda_1=2$ for $q\le 64$ and $4$ for $q\ge 128$ under the correlated Eve; $\lambda_2$ uses the adaptive rule of Section~\ref{sec:vpq}. Training uses Adam with learning rate $3{\times}10^{-5}$: $300$ epochs without or with an uncorrelated Eve, or $50$ pre-training epochs followed by adversarial fine-tuning up to epoch $1000$ under the correlated Eve. Reconciliation uses Reed--Solomon (RS) codes over $\GF(q)$ with blocklength $n=q-1$ and message length $m$ swept. Each test point uses $10^{6}$ samples. We compare VPQ at $q\in\{16,32,64,128,256\}$ against the three baselines of Section~\ref{sec:intro}: Conv.\ KG~\cite{zhang2016key} (CDF quantization and Gray coding); MIAE~\cite{guo2025physical}; and DAAE~\cite{zhou2023daae}.

\emph{Convergence} (Fig.~\ref{fig:conv}). $H(W)$ approaches $\log q$ within tens of epochs in every scenario, and $\pr{W=V}$ stabilizes after the transition from pre-training to adversarial fine-tuning. The empirical gap $\calI_{\mathrm{VUB}}-\calI_{\mathrm{VLB}}$ settles near $\log q$, which by Remark~\ref{rem:lautum} is the signature of the zero-leakage regime.

\emph{Reed--Solomon reconciliation} (Fig.~\ref{fig:kmr}). At $\mathrm{SNR}_{\mathrm{ab}}=20$~\si{dB} we sweep the RS message length $m$ to trace the key mismatch rate (KMR) against the target key rate $R=(m/n)\log q$. The triangles mark $R^{\star}_{\mathrm{univ}}$ at the measured $p_e$; by Theorem~\ref{thm:univ}, no scheme that is reliable and secure uniformly over $\calD_{p_e}$ can operate to the right of the corresponding triangle. The RS cliffs sit strictly to the left of the triangles, which is the expected finite-blocklength gap of a length-$(q-1)$ code. At moderate $p_e$ (no-Eve and uncorrelated-Eve panels) a larger $q$ shifts both the triangle and the cliff to the right, so larger alphabets reach higher target rates. Under a correlated Eve the adversarial constraint inflates $p_e$ more for large $q$ and collapses $R^{\star}_{\mathrm{univ}}$, so $q=256$ already yields $\mathrm{KMR}\approx 1$ at $R=1$ and the ordering inverts: the smallest alphabets deliver the best low-rate operating point.

\emph{Secret key rate under a correlated Eve} (Fig.~\ref{fig:corr}). For each frozen encoder we retrain a fresh predictor $p_\psi$ and report $I_{\mathrm{VLB}}$ as a tight estimate of $I(W;Z)$. VPQ keeps $I(W;Z)\le 0.3$ bits per symbol across all $q$ and all tested SNRs, so by Proposition~\ref{prop:bound} the corresponding $C_{\mathrm{SK}}$ is directly realizable as a usable key without privacy amplification. The three baselines leak 1.5 to 2.2 bits per symbol at 20 dB and above.

\emph{Complexity.} The dominant per-step training cost is the cross-batch double sum in $\calI_{\mathrm{VUB}}$~\eqref{eq:vub}, $\mathcal{O}(B^2q)$ operations, or $1.1\times10^{9}$ at $B=2048$ and $q=256$; the sum is a single matrix product and occurs only during training. Encoder and predictor passes add $\mathcal{O}(B\,d_{\mathrm{FCN}}^2)$ per layer and the remaining losses $\mathcal{O}(Bq)$. At key generation the encoder runs one forward pass per observation, $7.6$~M multiply--accumulates at $q=256$, against a per-dimension CDF lookup for Conv.\ KG; its memory is the $7.6$~M weights, $31$~MB in single precision, independent of the key length. The autoencoder baselines likewise run one forward pass per observation, and the sketch stage is common to all four methods.

%% file: sections/conclusion.tex
\section{Conclusion}
\label{sec:conc}

We presented VPQ, a two-stage framework in which encoders trained by a variational adversarial objective map correlated observations into a discrete key alphabet and a linear code-offset secure sketch reconciles the outputs into an identical key. The VPQ losses correspond term by term to the lower bound on the one-way secret key capacity of the induced source (Theorem~\ref{thm:vpq}), and $q$-ary linear sketches attain the worst-case universal rate of Theorem~\ref{thm:univ} over the corresponding i.i.d.\ source class. On Gaussian fading channels under absent, uncorrelated, and correlated Eve, VPQ keeps the leakage to Eve at most $0.3$ bits per symbol against up to $2.2$ for the baselines, and Reed--Solomon reconciliation operates within the predicted finite-blocklength gap. The quantization stage is independent of the direction of the public discussion, so a two-way reconciliation protocol, in which both parties publish, replaces only the sketch.

%% file: refs.bib
@article{ahlswede1993common,
  title={Common randomness in information theory and cryptography. {I. Secret} sharing},
  author={Ahlswede, Rudolf and Csisz{\'a}r, Imre},
  journal={IEEE Trans. Inf. Theory},
  volume={39},
  number={4},
  pages={1121--1132},
  year={1993}
}

@article{maurer1993secret,
  title={Secret key agreement by public discussion from common information},
  author={Maurer, Ueli M.},
  journal={IEEE Trans. Inf. Theory},
  volume={39},
  number={3},
  pages={733--742},
  year={1993}
}

@book{csiszar2011information,
  title={Information Theory: {C}oding Theorems for Discrete Memoryless Systems},
  author={Csisz{\'a}r, Imre and K{\"o}rner, J{\'a}nos},
  edition={2},
  year={2011},
  publisher={Cambridge Univ. Press}
}

@inproceedings{dodis2004fuzzy,
  title={Fuzzy extractors: {H}ow to generate strong keys from biometrics and other noisy data},
  author={Dodis, Yevgeniy and Reyzin, Leonid and Smith, Adam},
  booktitle={Adv. Cryptol. -- {EUROCRYPT}},
  pages={523--540},
  year={2004}
}

@article{csiszar1982linear,
  title={Linear codes for sources and source networks: {E}rror exponents, universal coding},
  author={Csisz{\'a}r, Imre},
  journal={IEEE Trans. Inf. Theory},
  volume={28},
  number={4},
  pages={585--592},
  year={1982}
}

@book{roth2006introduction,
  title={Introduction to Coding Theory},
  author={Roth, Ron M.},
  year={2006},
  publisher={Cambridge Univ. Press}
}

@inproceedings{wang2011fast,
  title={Fast and scalable secret key generation exploiting channel phase randomness in wireless networks},
  author={Wang, Qian and Su, Hai and Ren, Kui and Kim, Kwangjo},
  booktitle={Proc. IEEE INFOCOM},
  pages={1422--1430},
  year={2011}
}

@article{zhang2016key,
  title={Key generation from wireless channels: {A} review},
  author={Zhang, Junqing and Duong, Trung Q. and Marshall, Alan and Woods, Roger},
  journal={IEEE Access},
  volume={4},
  pages={614--626},
  year={2016}
}

@article{nguyen2021security,
  title={Security and privacy for {6G}: {A} survey on prospective technologies and challenges},
  author={Nguyen, Van-Linh and Lin, Po-Ching and Cheng, Bo-Chao and Hwang, Ren-Hung and Lin, Ying-Dar},
  journal={IEEE Commun. Surv. Tut.},
  volume={23},
  number={4},
  pages={2384--2428},
  year={2021}
}

@inproceedings{han2020physical,
  title={Physical layer secret key generation based on autoencoder for weakly correlated channels},
  author={Han, Jingyuan and Zeng, Xin and Xue, Xiaoping and Ma, Jingxiao},
  booktitle={Proc. IEEE/CIC ICCC},
  pages={1220--1225},
  year={2020}
}

@article{guo2025physical,
  title={Physical layer secret key generation based on mutual information-driven autoencoder},
  author={Guo, Dengke and Xiong, Jun and Ma, Dongtang and Liu, Xiaoran and Wei, Jibo},
  journal={IEEE Trans. Wireless Commun.},
  volume={24},
  number={10},
  pages={8042--8056},
  year={2025}
}

@article{zhou2023daae,
  title={Physical-layer secret key generation based on domain-adversarial training of autoencoder for spatial correlated channels},
  author={Zhou, Qingjiang and Zeng, Kai},
  journal={Appl. Intell.},
  volume={53},
  pages={5304--5319},
  year={2023}
}

@inproceedings{poole2019variational,
  title={On variational bounds of mutual information},
  author={Poole, Ben and Ozair, Sherjil and {van den Oord}, A\"aron and Alemi, Alex and Tucker, George},
  booktitle={Proc. Int. Conf. Mach. Learn. ({ICML})},
  pages={5171--5180},
  year={2019}
}

@inproceedings{cheng2020club,
  title={{CLUB}: {A} contrastive log-ratio upper bound of mutual information},
  author={Cheng, Pengyu and Hao, Weituo and Dai, Shuyang and Liu, Jiachang and Gan, Zhe and Carin, Lawrence},
  booktitle={Proc. Int. Conf. Mach. Learn. ({ICML})},
  pages={1779--1788},
  year={2020}
}

@article{palomar2008lautum,
  title={Lautum information},
  author={Palomar, Daniel P. and Verd{\'u}, Sergio},
  journal={IEEE Trans. Inf. Theory},
  volume={54},
  number={3},
  pages={964--975},
  year={2008}
}

@article{wyner1975wire,
  title={The wire-tap channel},
  author={Wyner, Aaron D.},
  journal={Bell Syst. Tech. J.},
  volume={54},
  number={8},
  pages={1355--1387},
  year={1975}
}
